\documentclass[12pt]{article}

\usepackage{amsmath,amsfonts,amsthm,amscd,amssymb,slashed,url,doi}
\usepackage[pdftex]{graphicx}
\usepackage{algpseudocode}
\usepackage{tikz}
\usepackage{enumerate}
\usetikzlibrary{shapes,arrows}
\usetikzlibrary{matrix,arrows}
\usepackage{epstopdf}
\usepackage{slashed}
\usepackage{xcolor,import}
\usepackage{transparent}
\usepackage[all]{xy}
\usepackage{subcaption}

\usepackage{geometry}
\makeatletter

\@addtoreset{equation}{section}
\def\section{\@startsection {section}{1}{\z@}{-2.5ex plus -1ex minus
 -.2ex}{1.3ex plus .2ex}{\large\bf}}
\def\subsection{\@startsection{subsection}{2}{\z@}{-2.25ex plus%
 -1ex minus -.2ex}{0.5ex plus .2ex}{\bf}}

\newcommand{\R}{\mathbb{R}}

\newcommand{\C}{\mathbb{C}}
\newcommand{\Z}{\mathbb{Z}}

\theoremstyle{definition}

\def\d{\mathrm{d}}
\def\bepsilon{\boldsymbol{\epsilon}}
\def\bn{\boldsymbol{n}}
\def\be{\boldsymbol{e}}
\def\bF{\boldsymbol{F}}
\def\bA{\boldsymbol{A}}

\def\bee{\begin{equation}}
\def\eee{\end{equation}}
\numberwithin{equation}{section}

\newtheorem{theorem}{Theorem}[section]
\newtheorem{lemma}[theorem]{Lemma}

\begin{document}
\parskip 6pt
\parindent 0pt

\baselineskip 28pt

\begin{center}
{\Large \bf Magnetic Skyrmions at Critical Coupling}

\baselineskip 18pt

\vspace{0.4 cm}

{\bf Bruno Barton-Singer, Calum Ross and Bernd J~Schroers}\\
\vspace{0.2 cm}
Maxwell Institute for Mathematical Sciences and
Department of Mathematics,
\\Heriot-Watt University,
Edinburgh EH14 4AS, UK. \\
{\tt  bsb3@hw.ac.uk},  {\tt cdr1@hw.ac.uk}  and {\tt b.j.schroers@hw.ac.uk} \\

\vspace{0.4cm}

{ 27  November   2019} 
\end{center}

\begin{abstract}
\noindent We introduce a family of models for  magnetic skyrmions in the plane  for which infinitely many solutions can be given explicitly. 
The energy defining the models  is bounded below by a linear combination of degree and total vortex strength, and the configurations attaining the bound satisfy a first order Bogomol'nyi equation.  We give explicit solutions which depend on an  arbitrary holomorphic function. The simplest solutions are the  basic Bloch and N\'eel skyrmions, but we also exhibit distorted   skyrmions and anti-skyrmions as well as line defects and configurations consisting of skyrmions and anti-skyrmions.

\end{abstract}

\baselineskip 16pt
\parskip 4 pt
\parindent 10pt

\section{Introduction}
Magnetic skyrmions are topologically non-trivial configurations which occur in certain magnetic materials. It was first observed in \cite{BY} that particular examples of such  configurations are minimisers of  natural  energy expressions  for the magnetisation vector. They have since then become the subject of intense study, both experimentally and theoretically, not least because of their potential use as  information carriers in  magnetic  storage devices, see  \cite{NT} for a review.

In this paper  we use tools from gauge theory, complex analysis and differential geometry to introduce models for magnetic skyrmions which can be solved explicitly.  The models are grounded in the physics of magnetic skyrmions and 
belong to the general class   already considered in \cite{BY},  with an  energy expression consisting  of  a Dirichlet term, a Dzyaloshinskii-Moriya (DM) interaction energy \cite{Dzyaloshinskii,Moriya} and a potential  combining Zeeman  and  easy plane anisotropy terms.  However,  our formulation reveals  that, for   critical values of the coupling constants,  the models are of Bogomol'nyi type, which means that static solutions can be obtained by solving  a first order partial differential equation. Moreover, this equation can be solved in terms of an arbitrary holomorphic function.   Both the Bogomol'nyi property and the existence of an infinite family of explicit solutions are novel in the context of  magnetic skyrmions.

Models of Bogomol'nyi type have historically played an important role in the study of topological solitons \cite{MS}. They require a particular choice of coupling constants, but their mathematical properties allow for a far more detailed and explicit study of the solitons and their dynamics than would be possible in the generic case. Intricate  and surprising features of soliton dynamics such as shapes  and symmetries of multi-soliton configurations or scattering behaviour were first observed in models of Bogomol'nyi type and later found in generic soliton models.

The focus of this paper is the mathematical structure of the critically coupled models for magnetic skyrmions,  and we only begin  to explore the properties of our solutions.  However,  even at this stage it is clear that our infinite family of solutions   contains several of the configurations associated with the various phases of  generic models \cite{latticeground1,latticeground2,phasediagram}, and that it illustrates elliptical deformations  \cite{elliptical_first, elliptical_easy} and  the recently studied skyrmion bags \cite{FKATDS}  or `sacks' \cite{RK}.

The simplest  topological soliton  theory of Bogomol'nyi type is the  $O(3)$-sigma model in the plane \cite{BP}. The basic field is a map from the plane to the sphere, and  finiteness of the Dirichlet energy requires the field to tend to  a constant at spatial infinity, and  to extend to a map from sphere to sphere. Such maps have a topological and integer degree, which is physically interpreted as the soliton number. The energy is bounded below by a multiple of the absolute value of the degree, and this bound is attained by configurations which satisfy the  first order Bogomol'nyi equation. In this particular case, the Bogomol'nyi equation requires the configuration  to be a holomorphic or anti-holomorphic map to the Riemann sphere.

In analogy with the baby skyrme model \cite{PSZ},  one would expect the inclusion of DM interactions, Zeeman potential and anisotropy terms inevitably  to destroy the Bogomol'nyi property of the pure $O(3)$ sigma model. However, here we shall show that, with a careful choice of potential and for a one-parameter family of DM interaction terms, our models preserve the Bogomol'nyi property  and can be solved in terms of a fixed anti-holomorphic and an arbitrary holomorphic map to the Riemann sphere.  The fixed anti-holomorphic part turns out to be an analytical version of the usual Bloch  or  N\'eel magnetic skyrmions, but the holomorphic part is new.

Our family of models is introduced in Sect.~2, and the solutions are studied in Sect.~5. Readers primarily interested in the models and their solutions  are invited to skip directly  from Sect.~2 to Sect.~5. In the intervening sections we derive the Bogomol'nyi equation in two different ways. In Sect.~3, we write  the theory as a non-abelian gauge theory with a fixed non-abelian gauge field and apply a trick for constructing gauged sigma models of Bogomol'nyi type  introduced in \cite{Schroers}. In Sect.~4, we derive the Bogomol'nyi equation in complex stereographic coordinates and give the general  solution in terms of fixed anti-holomorphic and an arbitrary holomorphic function. We show that,  when that holomorphic function is rational,  the energy is generically  positive and quantised in multiples of $4\pi$.  We also point out that the energy is not well-defined when the leading holomorphic term at spatial infinity is linear, and propose a regularisation which preserves the generic formula.
We  study detailed properties of rational solutions in Sect.~5. The final Sect.~6 contains our conclusion and a brief discussion of open questions.

\section{The model}
\subsection{Energy and symmetry}

The basic field in any  mathematical model for magnetic skyrmions is the magnetisation vector, which in the planar and static case is a map  $\bn:\R^2 \rightarrow S^2$. Here we consider models where the energy of a configuration is measured by the sum of three terms: the Dirichlet or Heisenberg energy (quadratic in derivatives), a  generalised DM  interaction energy (linear in derivatives), and a potential term which may involve linear or quadratic terms in the Cartesian components of $\bn=(n_1,n_2,n_3)^t$ (with the constraint    $n_1^2+n_2^2+n_3^2=1$ always understood). General models of this sort were considered in   the seminal paper \cite{BY} in which the possibility  of topologically  stable configurations now known as magnetic skyrmions was first pointed out, and have been widely studied since then.

Our one-parameter  family of models belongs to the general family considered in \cite{BY} but requires a particular choice of coupling constants, which we call critical. The parameter in the family is an angle $\alpha \in [0,2\pi)$, and describes a generalised DM interaction term. In order to define this term, we need some additional notation. 

We use Cartesian coordinates $x_1$ and $x_2$ in the plane and write  $\partial_1$ and  $\partial_2$ for partial derivatives with respect to them.
Thinking of $\R^2$ as embedded in Euclidean $\R^3$ and using three-dimensional notation, we also write $\be_1, \be_2,\be_3$ for the canonical basis of $\R^3$, with $\be_3=\be_1\times \be_2$. 
In terms of the  rotation $ R(\alpha)$ about the 3-axis by $\alpha\in [0,2\pi)$,
we define 
\bee
\label{rotbasis}
\be_1^\alpha = R(\alpha) \be_1 = \begin{pmatrix} \cos\alpha \\ \sin\alpha \\ 0 
\end{pmatrix},
 \qquad \be_2^\alpha =  R(\alpha) \be_2 =  \begin{pmatrix} -\sin \alpha \\ \phantom{-} \cos\alpha \\ 0 
\end{pmatrix}, 
\qquad \be_3= \begin{pmatrix} 0 \\ 0 \\ 1
\end{pmatrix}.
\eee
Extending the usual definition of the gradient $\nabla = \sum_{i=1}^2 \be_i\partial_i$ to write 
\bee
\nabla_\alpha =\sum_{i=1}^2  \be_i^\alpha \partial_i, 
\eee
and defining 
\bee
\label{alphandef}
\bn^\alpha= R(\alpha) \bn,
\eee
 the family of  DM interaction terms we  are interested in  is 
\bee
\label{ourDM}
 \bn^\alpha \cdot \nabla \times \bn^\alpha = \bn \cdot \nabla_{-\alpha}  \times \bn.  
\eee
In components, it  consist of two familiar parts:
\bee
\bn \cdot \nabla_{-\alpha}  \times \bn= \cos\alpha \; w_B +\sin\alpha  \; w_N, 
\eee
where  $w_b$ and $w_N$ are the following  contractions of  the  chirality  tensor $\bn \times \partial_i \bn$:
\begin{align}
w_B& = \phantom{-} n_1\partial_2n_3 -n_2\partial_1 n_3 + n_3(\partial_1n_2 -\partial_2n_1),\nonumber \\ 
w_N& = - n_1\partial_1n_3 -n_2\partial_2 n_3 + n_3(\partial_1n_1 +\partial_2n_2).
\end{align}

Choosing energy units so that the coefficient of the Dirichlet energy term is unity,  the family of energy  functionals we want to consider is 
\bee
\label{alphaen1}
E[\bn]=\int_{\R^2} \frac{1}{2}(\nabla \bn)^2 + \kappa\bn \cdot \nabla_{-\alpha}  \times \bn+  \frac{\kappa^2}{2}(1-n_3)^2  \; \d x_1\d x_2.
\eee
The coupling constant $\kappa$ could  also be set to unity by a choice of length unit, but we find it convenient to keep it in our calculations.  We assume $\kappa >0$ for the remainder of this paper.
Since 
\bee
\label{critpot}
 \frac 12 (1-n_3)^2 = (1-n_3) - \frac 12 (1-n_3^2),
\eee
our potential term  combines a Zeeman and easy plane anisotropy  potential, with coefficients chosen   in such a way that the sum is a perfect square.

The energy expression \eqref{alphaen1} is invariant under translations in the plane and under the group $O(2)$ of rotations and reflection of   the plane, combined with simultaneous rotations and reflections  of the target sphere. On fields, the rotations  act as  
\bee
\label{rotact}
\bn(x_1,x_2) \mapsto R(\sigma) \bn(\cos\sigma \, x_1 - \sin \sigma\, x_2, \sin \sigma \,x_1 + \cos\sigma \, x_2),  \quad \sigma \in [0,2\pi),
\eee
 and  the  generator of reflections acts as 
 \bee
 \label{reflectact}
 \bn(x_1,x_2) \mapsto R(2\gamma) \bar\bn( x_1, -x_2),  \quad \text{with} \quad \bar \bn = \begin{pmatrix}
\phantom{-} n_1 \\-n_2\\ \phantom{-} n_3 \end{pmatrix}\quad \text{and}  \quad \gamma= \frac{\pi} {2} -\alpha.
 \eee
The symmetry group is smaller than that of generic baby skyrme models \cite{PSZ} because the DM interaction breaks the product of the  orthogonal groups in space and target space to a  diagonal subgroup. It is worth noting that it would be  mathematically natural to   consider an alternative   version of the DM interaction with the opposite chirality:
\bee
\label{reflectedDM}
\bar \bn \cdot \nabla_{-\alpha}  \times  \bar \bn.
\eee
If this term was used instead of the standard DM interaction, the symmetry would consist of spatial rotations as in \eqref{rotact} but   combined with rotations $R(-\sigma)$ of $\bn$.  Both the DM interaction  term \eqref{ourDM} and the  flipped version \eqref{reflectedDM} were considered  in \cite{genDMI}, together with a generic linear combination of the Zeeman and anisotropy potential $(1-n_3^2)$. We will work with the  DM interaction \eqref{ourDM}  and primarily consider  the critical linear combination \eqref{critpot}, but we  discuss the more  general potential in Sect.~2.2  and comment  on how our results would change if we had used the opposite chirality in the Conclusion.

At this point we should really  specify which  boundary conditions we impose on the field $\bn$ at spatial infinity.  It is  {\em a priori} not clear if one should, as in the discussion of the $O(3)$ sigma model, consider only configurations which extend to continuous maps $S^2\rightarrow S^2$. If they did, then 
\bee
\label{Qdef}
Q[\bn]=\frac{1}{4\pi}\int_{\R^2} \bn \cdot \partial_1\bn \times \partial_2\bn \; \d x_1 \d x_2
\eee
would automatically be an integer, giving the degree of the extended map.

As we shall see, we  should in fact  allow for configurations which  do not have a continuous extension $S^2\rightarrow S^2$. Such maps do not have a topological degree, but the integral  expression for $Q$ still plays an important role.  We will refer to it as the degree throughout this paper. 

In our model, the degree occurs invariably  in conjunction with a term which also depends on the boundary behaviour, namely the total
 vortex strength
\bee
\label{Omdef}
\Omega[\bn]= \frac{1}{4\pi}\int_{\R^2} \omega \; \d x_1 \d x_2,
\eee
where the integrand is the vorticity of the  first two components of  $\bn^\alpha$:
\bee
\label{vorticity}
\omega = \kappa( \partial_1n_2^\alpha - \partial_2n_1^\alpha). 
\eee

The  expressions we have given for the total energy, the degree and the total vortex strength should all be interpreted as functionals on the space of magnetisation fields. For some magnetisation fields $\bn$, the relevant integrals   may not be well-defined. One might  therefore want to restrict the following discussion to a class of configurations which  have a well-defined total energy, total vortex strength and degree. However, as we shall see, it is impossible to do this without discarding some of the most interesting configurations which arise as solutions in our model. In order to  keep the discussion general but also  mathematically rigorous, we therefore  need notation for the  restrictions of the integrals \eqref{alphaen1},\eqref{Qdef} and \eqref{Omdef} to  compact subsets $D\subset \R^2$. We define
\begin{align}
\label{EQOmD}
E_D[\bn]&=\int_{D } \frac{1}{2}(\nabla \bn)^2 + \kappa\bn \cdot \nabla_{-\alpha}  \times \bn+  \frac{\kappa^2}{2}(1-n_3)^2  \; \d x_1\d x_2, \nonumber \\
Q_D[\bn]&=\frac{1}{4\pi}\int_{D} \bn \cdot \partial_1\bn \times \partial_2\bn \; \d x_1 \d x_2, \nonumber \\
\Omega_D[\bn]&= \frac{1}{4\pi}\int_D \omega \; \d x_1 \d x_2, \qquad \qquad \qquad D\subset \R^2.
\end{align}

Clearly, $\omega \, \d x_1\wedge \d x_2 =\d \Theta$, where
\bee
\label{Thetadef}
\Theta =\kappa(n_1^\alpha \d x_1+ n_2^\alpha \d x_2)
\eee
is a differential one-form which plays a central role in this paper. It then follows that 
\bee
\label{OmdefD}
\Omega_D[\bn]= \frac{1}{4\pi}\int_{\partial D} \Theta.
\eee
In particular, one can take $D$ to be a disk of radius $R$ and centred at the origin, so that $\partial D= C^R$ is the circle of radius $R$.  For  some of the configurations we consider,  the limit 
\bee
\label{omthetareg}
\Omega^\circ[\bn] = \frac{1}{4\pi}\lim_{R\rightarrow \infty} \int_{C^R} \Theta
\eee
exists  even when the  integral defining $\Omega[\bn]$ does not. We will treat $\Omega^\circ[\bn]$ as a regularised total vortex strength in those  cases. When $\Omega$ is well-defined  it necessarily coincides with $\Omega^\circ$. 

For the solutions we construct in this paper, the total  vortex strength is generically well-defined and finite, and, as a consequence, the total energy turns out to be  given by a simple formula. The regularisation \eqref{omthetareg} is  such that the resulting regularised  energy for the non-generic solutions naturally fits into this general formula. However, we should stress that the regularisation procedure is not  essential  for our main results.\footnote{See the  Note added at the end of this paper for a  brief discussion of an alternative energy expression which differs  from the  one defined in \eqref{EQOmD} by a boundary term and leads to the same variational equations.}.

Postponing a  more detailed discussion of allowed configurations and topological invariants to Sect.~4 and the Conclusion, we now derive the variational equation for \eqref{alphaen1}, only assuming that $\bn$ is twice differentiable. 
By considering the variation 
$\delta \bn = \bepsilon \times \bn$ 
for an infinitesimal vector function $\bepsilon$ which vanishes rapidly at spatial infinity, we obtain 
\bee
\label{eomm}
2\kappa (\bn \cdot \nabla_{-\alpha}) \bn = \left( \Delta \bn + \kappa^2 (1-n_3)\be_3\right) \times \bn.
\eee 
We will show that this equation is  in fact implied by  a first order equation. 

\subsection{Hedgehog fields}
In the skyrmion literature, the magnetisation $\bn$ is often described in terms of spherical polar coordinates, defined via
\bee
\label{nangle}
\bn=\begin{pmatrix}
\sin \theta \cos \phi \\
\sin\theta \sin \phi \\ 
\cos \theta
\end{pmatrix},
\eee
where $\theta$ and $\phi$ are functions on the plane. This parametrisation is particularly useful when considering hedgehog fields. 
 By definition, and using polar coordinates  $(r,\varphi)$  in the plane,   hedgehog fields have a  profile $\theta$ which depends on $r$ only and a longitudinal angle $\phi$ which is related to $\varphi$  according to
 \bee
 \label{hedgeangle}
 \phi = \varphi + \gamma,
 \eee
for a constant angle $\gamma$.   Such fields are invariant under the rotational  symmetry \eqref{rotact}. They are additionally invariant under the reflection symmetry \eqref{reflectact} if and only if we choose $\gamma$ to be the complementary angle of $\alpha$ as in \eqref{reflectact},  and we now make this choice.  
With the boundary condition 
\bee
\label{boundary}
\theta(0)=\pi, \qquad \theta(\infty) = 0,
\eee
one checks that hedgehog fields have degree $Q=-1$. 

Before developing the general machinery for generating solutions of the equation \eqref{eomm}, we note some properties of the much simpler hedgehog solutions in our model. We do this in a slightly more general family of models, obtained from 
 \eqref{alphaen1}  by replacing 
\bee
\label{potmu}
\frac{\kappa^2}{2}(1-n_3)^2  \rightarrow \frac{\mu^2}{2}(1-n_3)^2
\eee
for a further real constant $\mu$. Minimisers of the resulting energy functional were studied in  \cite{DM}, and those of degree $Q=-1$  were shown to have holomorphicity properties similar to the ones which we will  demonstrate more generally for stationary points of \eqref{alphaen1}. We will exhibit these  properties in   our discussion of example solutions in Sect.~5. 
Here we  derive the profile of  hedgehog solutions by a more pedestrian method. 

For hedgehog fields and $\gamma= \frac \pi 2-\alpha$, the energy expression with the replacement \eqref{potmu} is 
\bee
E=2\pi \int_0^\infty rdr \left(\frac 12 \left(\frac{d \theta}{dr}\right)^2 + \frac{\sin^2\theta }{2 r^2} +\kappa   \left(\frac{d\theta}{dr} +\frac{\sin(2\theta)}{2r}\right) +\frac{\mu^2}{2}(1-\cos\theta )^2\right).
\eee
The Euler-Lagrange equation is
\bee
\label{newEL}
\frac{d^2\theta}{dr^2}  =-  \frac 1 r \frac{d\theta}{dr} + 
 \frac{\sin(2\theta)}{2r^2} - 2 \kappa \frac{\sin^2 \theta}{r}  +\mu^2 \sin\theta(1-\cos\theta). 
\eee
With the boundary condition \eqref{boundary}, this  is  solved by
 \bee
\label{hedgehoggen}
\theta = 2\tan^{-1}\left(\frac{2\kappa }{\mu^2 r}\right).
\eee 
As already advertised, we will recover this profile from the simplest solution of a first order Bogomol'nyi equation  
in the case $\mu=\kappa$ in equation \eqref{hedgehog} .

\section{A Bogomol'nyi equation for magnetic skyrmions}

We will now show that the energy functional \eqref{alphaen1}  can be written as the  sum of a squared expression and a linear combination  of  the integral expression for the  degree \eqref{Qdef} and the total vortex strength \eqref{Omdef}. The vanishing of the squared expression gives  a first order Bogomol'nyi equation which implies the variational second order equation \eqref{eomm}.

Our derivation of the Bogomol'nyi equation is inspired by a similar treatment of gauged sigma models in \cite{Schroers} and \cite{Nardelli}. As noticed in \cite{Melcher}, the combination 
\bee
\partial_i\bn -\kappa \be_i\times \bn,  \qquad i=1,2,
\eee
which occurs in many calculations involving magnetic skyrmions and which is often called `helical derivative' can be thought of as a  covariant derivative with respect to a  non-abelian gauge field. To see the benefits of this, we take a more general viewpoint and consider more general  $su(2)$  gauge fields. 

To minimise notation, we identify  the $su(2)$ Lie algebra with $\R^3$ and  the Lie algebra commutator with the vector product. 
Defining the covariant derivative of $\bn$ as 
\bee
 D_i\bn = \partial_i \bn + \bA_i \times \bn,
\eee
and the non-abelian field strength
\bee
\bF_{ij}=\partial_i \bA_j - \partial_i \bA_j   
 +  \bA_i\times \bA_j, \qquad i,j=1,2,
\eee 
we note 
\bee
\label{trivga}
(D_1\bn + \bn\times D_2\bn)^2  =  (D_1\bn)^2 + (D_2\bn)^2  - 2D_1\bn\times D_2\bn \cdot \bn. 
\eee
and also, as already observed by 't Hooft \cite{tHooft},
\bee
\label{trickga} 
\bn \cdot D_1\bn \times D_2\bn -  \bn \cdot \bF_{12} = \bn \cdot \partial_1\bn\times \partial_2\bn + \partial_2(\bn\cdot \bA_1) - \partial_1(\bn\cdot \bA_2).
\eee
This equation  shows that the particular combination of the degree density (the integrand of \eqref{Qdef})  with a  two-dimensional curl on the right hand side can be expressed in a manifestly gauge invariant way.

We can now state and prove the main result in this section. 
\begin{lemma}
The energy  for magnetic skyrmions at critical coupling associated with a compact subset $D\subset \R^2$  can be written as 
\bee
\label{EQTheta}
E_D[\bn] = 4\pi (Q_D[\bn] + \Omega_D[\bn])+ \int_{D} (D_1\bn +\bn \times D_2\bn)^2 \; \d x_1\d x_2,
\eee
where we used the covariant derivative
\bee
D_i\bn = \partial_i\bn -\kappa \be_i^{-\alpha}\times \bn, \quad i=1,2,
\eee
defined in terms of \eqref{rotbasis}. In particular, the equality 
\bee
\label{EQThetaequal}
E_D[\bn] = 4\pi (Q_D[\bn] +\Omega_D[\bn])
\eee
 holds  for all compact subsets $D\subset \R$  iff the Bogomol'nyi equation
\bee
\label{alphabogon}
D_1\bn =-\bn\times D_2\bn
\eee
is satisfied. This equation implies the variational equation \eqref{eomm}.
\end{lemma}

\begin{proof} 
Consider the gauge field given by the  constant, Lie algebra-valued one-form with Cartesian components  
\bee
 \bA_i= -\kappa \be_i^{-\alpha}, \qquad i=1,2.  
\eee
Then
\bee
 \bF_{12}= \kappa^2 \be_3, \quad \bn\cdot \bA_i =-\kappa \bn_i^{\alpha},  \quad i=1,2,
\eee
and therefore combining the results \eqref{trivga} and \eqref{trickga} for this gauge field gives
\bee
(D_1\bn + \bn\times D_2\bn)^2  =  (D_1\bn)^2 + (D_2\bn)^2 - 2\left( \bn\cdot \partial_1\bn\times \partial_2\bn + \kappa  (\partial_1 n_2^\alpha-\partial_2n_1^\alpha) + \kappa^2n_3\right).
\eee
One also checks that 
\bee
\frac{1}{2}(D_1\bn^2 +D_2\bn^2)=
\frac{1}{2}(\nabla \bn)^2 + \kappa \bn\cdot \nabla_{-\alpha} \times \bn  
 +\frac 12 \kappa^2(1+n_3^2),
\eee
and so the energy density of \eqref{alphaen1} can be written as 
\begin{align}
\frac{1}{2}(\nabla \bn)^2 +& \kappa \bn \cdot \nabla_{-\alpha} \times \bn+ \frac{\kappa^2}{2}(1-n_3)^2  \nonumber \\
&=\frac 12 (D_1\bn + \bn\times D_2\bn)^2 +\bn\cdot \partial_1\bn\times \partial_2\bn +\kappa( \partial_1n_2^\alpha-\partial_2n_1^\alpha).
\end{align} 
Integrating and  using the definitions \eqref{EQOmD}, we deduce that  the  energy  associated with a compact subset $D\subset R^2$ 
can be written as claimed in \eqref{EQTheta}. The equality \eqref{EQThetaequal} holds  for all $D\subset \R^2$ iff 
\bee
\label{bothbogo}
D_1\bn =- \bn\times D_2\bn \Leftrightarrow D_2\bn = \bn\times D_1\bn,
\eee
where the equivalence follows by applying $\bn\times$.  

Showing that the equation \eqref{bothbogo}  implies the variational equation is a lengthy but standard calculation.  We indicate the main steps.
Spelling out the Bogomol'nyi equation, we have  
\begin{align}
    \partial_1\bn & = -\bn \times \partial_2 \bn + \kappa(\be_1^{-\alpha}\times \bn + \bn \times (\be^{-\alpha}_2\times \bn)), \nonumber \\
    \partial_2\bn & = \bn \times \partial_1 \bn + \kappa(\be^{-\alpha}_2\times \bn - \bn \times (\be_1^{-\alpha}\times \bn)).
\end{align}
Therefore 
\begin{align}
&\partial_1^2 \bn + \partial_2^2 \bn  = 2 \partial_2\bn\times \partial_1\bn 
+\kappa(\be_1^{-\alpha}\times \partial_1\bn + \be_2^{-\alpha}\times \partial_2\bn) \nonumber  \\
    &+  \kappa (\partial_1\bn \times (\be_2^{-\alpha}\times \bn)+  \bn \times (\be_2^{-\alpha}\times \partial_1 \bn)- \partial_2\bn \times (\be_1^{-\alpha}\times \bn))- \bn \times (\be_1^{-\alpha} \times \partial_2\bn)).
\end{align}
Taking a cross product with $\bn$ and noting
\bee 
\bn\times (\be_i^{-\alpha}\times \partial_j\bn) = -n_i^\alpha\partial_j\bn,
\eee
we arrive at 
\bee
\bn \times \Delta \bn = -\kappa(n_1^\alpha\partial_1 + n_2^\alpha\partial_2) \bn + 
\kappa( n_1^\alpha \bn\times \partial_2\bn - n_2^\alpha\bn\times \partial_1\bn).
\eee
Now we use the Bogomol'nyi equation again in the last term
to conclude that 
\bee
\kappa( n_1^\alpha \bn\times \partial_2\bn - n_2^\alpha\bn\times \partial_1\bn)
= -\kappa(n_1^\alpha\partial_1 + n_2^\alpha \partial_2) \bn +\kappa^2(1-n_3)\be_3\times \bn.  
\eee
Therefore  
\bee
\bn \times \Delta \bn = -2\kappa(n_1^\alpha\partial_1 + n_2^\alpha \partial_2) \bn + \kappa^2(1-n_3)\be_3\times \bn,
\eee
which is  the equation \eqref{eomm} obtained by variation. 
\end{proof}

As often in the $O(3)$ sigma model or its gauged versions, the  Bogomol'nyi equations  are best studied in complex, stereographic coordinates. We do this in the next section.

\section{Magnetic skyrmions in  complex coordinates }
\subsection{The Bogomol'nyi equation  in stereographic  coordinates}
We use stereographic coordinates on the sphere defined by projection from the south pole. With the abbreviation  
\bee
\nu = n_1+in_2,
\eee
our stereographic coordinate is 
\bee
w= \frac{\nu }{1+n_3}.
\eee
When the magnetisation tends to the minimum of the potential term,  $\bn \rightarrow (0,0,1)^t$, then $w\rightarrow 0$. This makes  $w$ a natural choice of coordinate,  but in describing our solutions we also need  
\bee
\label{vdef}
v =\frac{1}{w}. 
\eee
For later use we also note the inverse relation
\bee
\label{invstereo}
\nu = \frac{2w }{1+|w|^2} ,\quad n_3=  \frac{1-|w|^2 }{1+|w|^2}.
\eee
We introduce the complex coordinate $z=x_1+ix_2$ in the plane, and use the standard holomorphic and anti-holomorphic derivatives
\bee
\partial_z= \frac{1}{2} (\partial_1-i\partial_2), \qquad 
\partial_{\bar z} = \frac{1}{2} (\partial_1+i\partial_2). 
\eee

Observing that, with the  notation \eqref{alphandef}, $ e^{i\alpha} \nu  = n_1^\alpha + in_2^\alpha$, 
the DM interaction term \eqref{ourDM} can be written in stereographic coordinates as 
\bee
\label{DMIcomplex}
\kappa \bn \cdot \nabla_{-\alpha}  \times \bn  = 2\kappa\text{Im}(e^{i\alpha}( n_3\partial_z \nu - \nu\partial_z n_3))= 4\kappa \text{Im} \left ( e^{i\alpha} \frac{\partial_z w + w^2 \partial_z \bar{w}} {(1+|w|^2)^2}\right).
\eee
The other terms in the energy functional have standard expressions in stereographic coordinates, and 
so  the energy \eqref{alphaen1} is 
\bee
\label{alphaen2}
E[w]= \int_{\R^2} \frac{2|\nabla w|^2 + 4 \kappa \text{Im} (e^{i\alpha}(\partial_z w + w^2 \partial_z \bar{w})) + 2\kappa^2| w|^4}{(1+|w|^2)^2}\; \d x_1 \d x_2.
\eee

The integral \eqref{Qdef} defining the degree   is  
\bee
\label{Qcomplex}
Q[w]= \frac{ i}{2\pi} \int_{\R^2} \frac{\partial_1w \partial_2\bar{w} - \partial_2 w\partial_1\bar{w}}{(1+|w|^2)^2} \; \d x_1 dx_2,
\eee
while the vorticity \eqref{vorticity} is 
\bee
\label{omegacomplex}
\omega =  2\kappa \text{Im}(e^{i\alpha} \partial_z \nu) = 4\kappa \text{Im}\left(e^{i\alpha} \frac{\partial_z w-w^2 \partial_z \bar{w}}{(1+|w|^2)^2}\right).
\eee
The one-form $\Theta$ can  be written as 
\bee
\label{Thetaalpha}
\Theta= \kappa  \text{Re} (e^{-i\alpha} \bar{\nu}  \d z) =\kappa \frac{2\text{Re} (e^{-i\alpha} \bar{w} \d z)}{1+|w|^2} =\kappa \frac{2\text{Re} (e^{-i\alpha} v \d z)}{1+|v|^2}.
\eee
In the following we write $E_D[w], Q_D[w]$ and $\Omega_D[w]$ for the integrals \eqref{EQOmD}  over $D\subset \R^2$  with the integrands expressed in terms of the complex field $w$. 
We can now state the main result of this paper.

\begin{theorem}
The energy \eqref{EQOmD} associated with a compact subset $D\subset \R^2$  can be written as 
\bee
\label{complexsquare}
E_D [w] = 4\pi (Q_D[w] +\Omega_D[w]) + \int_{D}8 \frac{( \partial_{\bar z} w  - \frac i 2    \kappa e^{i\alpha}  w^2 )(\partial_z \bar{w } + \frac i 2    \kappa e^{-i\alpha} \bar{w}^2 )}{(1+|w|^2)^2} \, \d x_1 \d x_2.
\eee
The equality 
\bee
\label{complexequality}
E_D[w] = 4\pi (Q_D[w] +\Omega_D[w]) 
\eee
holds for all  compact $D\subset \R^2$  iff the  field $v$ defined in \eqref{vdef} satisfies  the Bogomol'nyi equation 
\bee
\label{holobog}
 \partial_{\bar z} v = -\frac{i}{2}  \kappa      e^{i\alpha}.
 \eee
The general solution is 
\bee
\label{msol}
v = -\frac{i}{2}  \kappa      e^{i\alpha }\bar{z} + f(z),
\eee
where $f$ is an arbitrary holomorphic map from the plane to the Riemann sphere.
\end{theorem}

\begin{proof}
Using the standard identity 
\bee
\partial_z \bar{ w} \partial_{\bar z} w = \frac 1 4 \left(|\partial_1w|^2 + |\partial_2 w|^2 - i (\partial_1w \partial_2\bar{w} - \partial_2 w\partial_1\bar{w})\right).
\eee
and the  expression for the  vorticity \eqref{omegacomplex} in complex coordinates, we have the following identities for the energy density
\begin{align*}
&\frac{2|\nabla w|^2 + 4 \kappa \text{Im} (e^{i\alpha}(\partial_z w + w^2 \partial_z \bar{w})) + 2\kappa^2| w|^4}{(1+|w|^2)^2} \\
&= \omega + \frac{2|\nabla w|^2 + 8 \kappa \text{Im} (e^{i\alpha} w^2 \partial_z \bar{w}) + 2\kappa^2| w|^4}{(1+|w|^2)^2} \\
&= \omega  + 2i \frac{\partial_1w \partial_2\bar{w} - \partial_2 w\partial_1\bar{w}}{(1+|w|^2)^2}
 +\frac{8 \partial_z \bar{ w} \bar{\partial} w+ 8 \kappa \text{Im} (e^{i\alpha} w^2 \partial_z \bar{w}) + 2\kappa^2| w|^4}{(1+|w|^2)^2} \\
&=\omega  + 2i \frac{\partial_1w \partial_2\bar{w} - \partial_2 w\partial_1\bar{w}}{(1+|w|^2)^2}+  8 \frac{( \partial_{\bar z} w  - \frac i 2    \kappa e^{i\alpha}  w^2 )(\partial_z \bar{w } + \frac i 2    \kappa e^{-i\alpha} \bar{w}^2 )}{(1+|w|^2)^2} .
\end{align*}
Integrating and using the expression \eqref{Qcomplex} yields  the claimed expression \eqref{complexsquare} for the energy.

It follows  immediately that the equality \eqref{complexequality} holds   for all compact  $D\subset \R^2$   iff 
\bee
\label{complexbogo}
 \partial_{\bar z} w  =\frac i 2    \kappa e^{i\alpha} w^2,
\eee
which is the Bogomol'nyi equation \eqref{alphabogon}   in complex coordinates.
With $v$ as defined, this  is equi\-valent to
\bee
 \partial_{\bar z} v  =-\frac{i}{2}\kappa e^{i\alpha},
\eee
whose  general solution is $
v =  -\frac{i}{2}\kappa e^{i\alpha}
\bar{z} + f(z)$, 
where $f$ is an arbitrary holomorphic function, as claimed. Since $f$ takes values in the Riemann sphere, it is allowed to take the value $\infty$.
\end{proof}

One checks that the equation \eqref{holobog} is equivalent to the Bogomol'nyi equation \eqref{alphabogon}, and that it  therefore implies the variational equation \eqref{eomm}.  For later use we note that the energy density of configurations which satisfy the Bogomol'nyi equation is the vorticity plus  $4\pi$ times the integrand of the degree. This sum can be written as 
\bee
\label{endensityw}
\epsilon(x_1,x_2)= 4 \frac{ \partial_z w\partial_{\bar z} \bar w - \partial_{\bar z} w \partial_z \bar{w} +\kappa \text{Im}\left(e^{i\alpha} (\partial_z w -w^2\partial_z\bar w)\right)}{(1+|w|^2)^2}.
\eee 

\subsection{Degree and vorticity of magnetic skyrmions at critical coupling}

Before discussing the topology of the magnetic skyrmions defined by \eqref{msol}, it is worth revisiting the simpler case of the standard $O(3)$ sigma model in the plane, defined by the Dirichlet energy functional \cite{BP,MS}. The requirement of finite energy in that model  leads to the condition that fields tend to a constant at spatial infinity and may  be extended to smooth maps $S^2\rightarrow S^2$.  The Bogomol'nyi equations   are then equivalent to the map being either holomorphic or anti-holomorphic.   Considering the holomorphic case for definiteness,  the energy  is proportional to the degree and for this to be finite, the configuration has to be a rational map, i.e., of the form $p(z)/q(z)$, where $p$ and $q$ are polynomials of degree $m$ and $n$. The topological degree of the map is simply max$(m,n)$ in that case. 

The DM  term, which is a crucial feature of all models of magnetic skyrmions, is not positive definite, and therefore the energy expression for magnetic skyrmions may be finite even for configurations which do not tend to a constant value at spatial infinity.  As a result, even finite energy configurations do not necessarily extend to smooth maps $S^2\rightarrow S^2$ and do not necessarily have a well-defined topological degree.    Moreover, it is not clear {\em a priori } if   solutions of the Bogomol'nyi equation in our models have well-defined total vortex strength and total energy. 

We shall  now illustrate these issues  for our infinite family of solutions \eqref{msol}, and  show that, for rational maps $f$, the combination  $4\pi (Q+\Omega^0) $    of degree and  the regularised total vortex strength  \eqref{omthetareg} nonetheless  always yields a positive integer multiple of $4\pi$.

Before we enter a general discussion, it is illuminating to consider linear examples of the  form
\bee
\label{keyex}
v= -\frac i 2 \kappa e^{i\alpha} (\bar{z} +  A e^{i\chi}z), \qquad A\in\R^{\geq 0}.
 \eee
As we shall see, this family   captures the essence of the problem of defining the degree and the total vortex strength.

The evaluation of the integral defining $Q$  is elementary. Switching to polar  coordinates according to $z=re^{i\varphi}$, we find, after completing the radial integration, 
\bee
Q[w]= \frac{1}{2\pi}
\int_0^{2\pi} \frac{A^2 - 1}{1+ A^2 +2A\cos (2\varphi +\chi)}\d \varphi.
\eee

The evaluation of the total vortex strength is more subtle. The one-form $\Theta$ for the field \eqref{keyex} is 
\bee
\Theta = 4 \frac{ (1+ A\cos (2\varphi + \chi)) \d \varphi +  A\sin(2\varphi +\chi) \d \ln r}{r^{-2} + ( 1 + A^2  +2A \cos (2\varphi +\chi)}.
\eee
When  computing the total vortex strength, we need to integrate  this  form  over a  curve along which $r$ is large.  The leading term is 
\bee
\label{Thetasy}
\Theta \sim 4 \frac{ (1+ A\cos (2\varphi + \chi)) }{  1 + A^2  +2A \cos (2\varphi +\chi)   } \d \varphi +4\frac{A\sin(2\varphi +\chi)}{ 1 + A^2  +2A \cos (2\varphi +\chi)} \d \ln r.
\eee
Clearly, the integral  of the term proportional to $\d \varphi$ gives the same answer for any simple curve enclosing the origin.  However, the integral of the  term proportional to $ A \; \d \ln r$ depends on the curve we choose, even in the limit of large radius.  One can use the  $\varphi$-dependence  to introduce arbitrary contributions by  deforming the contour with an outward bulge starting at some angle $\varphi_1$ and ending at $\varphi_2  > \varphi_1$. 
We conclude that the total vortex strength is not well-defined for configurations defined by \eqref{keyex} when $A\neq 0$.

However,  the  integral of $\Theta$  along a large circle $C^R$ centred at the origin  has a well-defined limit as  the radius  tends to infinity, precisely because  $\d \ln r$  does not contribute along such a circle.  Thus, with the definition \eqref{omthetareg}
\bee
\Omega^0[w] =\frac{1}{2\pi} \int_0^{2\pi}\frac{2+2 A\cos(2\varphi+\chi)}{1 + A^2  +2A \cos (2\varphi +\chi)} \d \varphi. 
\eee
It is immediate that 
\bee
Q[w] + \Omega^\circ[w] =1,
\eee
regardless of the value of $A$. However, the  contribution from the degree and the vortex strength depends crucially on $A$. Since 
\bee
\label{Qformula}
 Q[w]= \frac{1}{2\pi} \int_0^{2\pi} \frac{A^2 - 1}{1+ A^2 + 2 A\cos (2\varphi +\chi)}\d \varphi =\begin{cases}
 \phantom{-}1 \quad & \text{if} \quad A>1  \\
 \phantom{-} 0 \quad & \text{if} \quad A=1 \\
 -1 \quad &\text{if} \quad A<1,
 \end{cases}
\eee
and 
\bee
\label{Omforula}
\Omega^\circ[w]=  \frac{1}{2\pi} \int_0^{2\pi} \frac{2+ 2A\cos(2\varphi+\chi)}{1 + A^2 +2 A\cos (2\varphi +\chi)}\d \varphi =\begin{cases} 
 0 \quad & \text{if} \quad A>1 \\
 1\quad & \text{if} \quad A=1  \\
 2 \quad &\text{if} \quad A< 1,
 \end{cases}
\eee
we see that a configuration dominated by the holomorphic part  ($A>1$) has degree 1 and  vanishing vortex strength. A configuration dominated by the anti-holomorphic part ($A<1$) has degree -1 but vortex strength 2. In the intermediate case $A=1$,  the degree comes out as 0  and the  vortex strength contributes 1. 

The deeper reason behind the `jumping' of the degree of the map defined by \eqref{keyex} lies in the extendibility of the map to  one between spheres. An overall factor is irrelevant for this discussion, so we consider 
\bee
v=\bar z +Ae^{i\chi} z.
\eee
Then $w=1/v$ has a pole when   $re^{-2i\varphi}= -Are^{i\chi }$. This has no solution when $A\neq 1$, but is solved by the entire line 
\bee
\label{linedirection}
\varphi= -\frac \chi 2  \pm \frac{\pi}{2}
\eee
 when $A=1$. In particular, $w$ therefore does not have a good limit for $z\rightarrow \infty$ when $A=1$: the  result is infinity along the direction \eqref{linedirection} but zero otherwise. It therefore  does not extend to a smooth map between spheres in  that case. When $A\neq 1$ one checks, by considering the map in terms of $\zeta =1/z$, that $w$ does extend to a smooth map between spheres. Our integrations confirm this analysis for $A\neq 1$, but also show that the combination of degree and vortex strength gives a stable result even when $A=1$.

Our observations about the examples \eqref{keyex} generalise. In order to formulate this generalisation we define the regularised energy of a solution of the Bogomol'nyi equation as 
\bee
E^\circ[w]= 4\pi (Q[w] + \Omega^0[w]).
\eee

\begin{lemma} If $p$ and $q$ are polynomials in $z$  of degree $m$ and $n$ and without common factor,  the  integral defining the  total energy of the magnetic skyrmion solution determined via
\bee
\label{findeg}
v =  -\frac{i}{2}\kappa e^{i\alpha}\bar{z} + \frac{p(z)}{q(z)}, 
\eee
is well-defined provided $m\neq n+1$, i.e., provided $p/q$ does not grow linearly for large $z$. The total energy is 
\bee
E[w]=4\pi \; \text{max} (m,n+1)\qquad  \text{if} \quad m\neq n+1.
\eee
When $m=n+1$, the total energy is not well-defined but the  regularised total energy  is 
\bee
E^\circ [w]=4\pi \; m \qquad   \text{if} \quad m=n+1.
\eee
\end{lemma}

\begin{proof}
It is clear that 
\bee
\label{ratf}
f(z)= \frac{p(z)}{q(z)} 
\eee 
is a  holomorphic map to the Riemann sphere, so \eqref{findeg} defines a magnetic skyrmion satisfying the Bogomol'nyi equation. 
Written in terms of $v$, the energy density  \eqref{endensityw}
is 
\bee
\label{endensity}
\epsilon(x_1,x_2)= 4 \frac{ \partial_z v\partial_{\bar z} \bar v - \partial_{\bar z} v \partial_z \bar{v} +\kappa \text{Im}\left(e^{i\alpha} (\partial_z \bar{v} -\bar {v}^2\partial_z v)\right)}{(1+|v|^2)^2}.
\eee 
This expression shows in particular that the energy density is smooth at the poles of $w$ (zeros of $v$), so that any divergence  in the energy integral must come from the behaviour at infinity. We first show that  there is no divergence when $m\neq n+1$.

For solutions of the form  \eqref{findeg} and $m-n> 1$, the leading  term in the energy density for large $r$  comes from 
\bee
4 \kappa\frac{  \text{Im}\left(-e^{i\alpha} \bar {v}^2\partial_z v\right)}{(1+|v|^2)^2},
\eee 
leading to the asymptotic formula  
\bee
|\epsilon(r,\varphi)|=  C r^{-(m-n+1)} + \mathcal{O} \left(r^{-(m-n+2)}\right), \quad \text{for some} \; C\in \R.
\eee
This is integrable with respect to the integration measure $r  \d r \, \d \varphi $  for $m-n> 1$. 

When $m -n < 1$, the leading large-$r$ behaviour in the energy density is determined by the linear anti-holomorphic term in $v$, so the  energy density behaves asymptotically as 
\bee
|\epsilon(r,\varphi)|=  C r^{-4} + \mathcal{O} \left(r^{-5}\right), \quad \text{for some} \; C\in \R.
\eee
This is again integrable with respect to the integration measure $r \d r \, \d \varphi $. 

Next we turn to the evaluation of the energy integral.  Our calculations will also show  that, for $m=n+1$, the integral  requires regularisation.  Our strategy is as follows. For any solution of the Bogomol'nyi equation we have $ E_D[w]=4\pi (Q_D[w] +\Omega_D[w])$ for any compact subset $D\subset \R^2$. 
In order to evaluate the energy integral, we  turn the integral defining the degree into boundary integrals and  evaluate them together with the boundary integral defining the total vortex strength. In the cases where the total energy is well-defined, we will find that the boundary contribution from infinity is independent of the choice of contour  at infinity.  In the case where $p/q$  grows linearly at infinity,  we evaluate  the boundary contribution on a circle at infinity,  leading to our  result for the regularised energy.

We recall  that the  expression \eqref{Qdef} for the degree  is the integral of the pull-back of the area form on $S^2$,  and that it can be written in terms of $w$  and  $v$ as 
\bee
4\pi Q[w] =2i\int_{\R^2} \frac{\d w\wedge \d \bar{w} }{(1+|w|^2)^2} =2i\int_{\R^2} \frac{\d v\wedge \d \bar{v} }{(1+|v|^2)^2}.
\eee
Moreover, the integrand can be written as an exact form
\bee
\label{degprim}
 2i\frac{\d v\wedge \d \bar{v} }{(1+|v|^2)^2}  =\d\left( i\frac{v\d\bar{v} - \bar{v} \d v}{1+|v|^2} \right) ,
\eee
but the one-form  in brackets is singular at the poles of $v$.  We can make these singularities explicit as follows.

With $f$ of the given form, we note 
\bee
v=\frac{g}{q},  \quad \text{with} \quad  g= -\frac{i}{2}\kappa e^{i\alpha}\bar z q(z) +p(z).
\eee
 For any $v$ of this form,
one checks that 
\bee
i\frac{v \d\bar{v} - \bar{v} \d v}{1+|v|^2} = i\frac{g \d \bar g -\bar g \d g + q \d\bar q -\bar q\d q}{|g|^2 +|q|^2} +i( \d\ln q-\d\ln\bar q).
\eee
The first term  on the right hand side is manifestly smooth, but  $
i( \d\ln q-\d\ln \bar q) $
is singular at the zeros of $q$. Thus, picking a compact region $D\subset \R^2$ which contains open neighbourhoods of the zeros of $q$, and denoting negatively oriented circles of radius $\epsilon$ around each of the zeros of $q$ (possibly repeated) by $C_j^\epsilon$, $j=1,\ldots,n$,  we can write 
\begin{align}
2i\int_{D} \frac{\d v\wedge \d \bar{v} }{(1+|v|^2)^2} &  = i \lim_{\epsilon\rightarrow 0} \sum_{i=1}^n\int_{C^\epsilon_j}(\d \ln q-\d\ln \bar q)
+ i \int_{\partial D} \frac{v\d\bar{v} - \bar{v}  \d v}{1+|v|^2} \nonumber \\
&= 4\pi n +  i \int_{\partial D} \frac{v\d\bar{v} - \bar{v} \d v}{1+|v|^2}.
\end{align}
Therefore, the degree and  the total vortex strength associated with the region $D$  can be combined into
\bee
\label{boundaryint}
4\pi (Q_D[w] + \Omega_D[w])= 4\pi n+  \int_{\partial D} \beta,
\eee
where we introduced the one-form
\bee
\beta = \frac{iv\d\bar{v} - i \bar{v} \d v+ \kappa e^{-i\alpha} v \d z + \kappa  e^{i\alpha} \bar{v} \d \bar{z} }{1+|v|^2},
\eee
which combines the one-form  whose exterior derivative is the degree density \eqref{degprim} with  the form $\Theta$ 
\eqref{Thetaalpha} used in the definition of the total vortex strength. 

Now, for solutions of the Bogomol'nyi equation,
\bee
\d v = -\frac{i}{2}\kappa e^{i\alpha} \d\bar{z} + \d f.
\eee
It follows that 
\bee
\beta = \frac{iv\d\bar{f} - i \bar{v} \d f+ \frac{\kappa}{2}e^{-i\alpha} v \d z + \frac{\kappa}{2} e^{i\alpha} \bar{v} \d\bar{z} }{1+|v|^2}. 
\eee
In order to evaluate the  integral  in \eqref{boundaryint} and its  limit,  we distinguish cases.

\noindent  (i) \;  $m > n+1$.\; In this case the leading term in $v$ for large $r$ is $az^{m-n}$ for some complex number $a$, and  the leading term for $f$ is also $az^{m-n}$. Inserting these, we find
\bee
 \beta \sim i(m-n)  (\d\ln \bar{z} -\d\ln z) = 2(m-n)\d \varphi.
\eee
The integral  of  the asymptotic form of $\beta $ around any simple curve   enclosing the origin is $4\pi(m-n)$, and we conclude 
\bee
4\pi (Q[w] + \Omega[w]) = 4\pi (n  +  (m-n)) = 4\pi m \quad \text{if} \;\;  m>n+1.
\eee
\noindent  (ii) \;  $m = n+1$.\; In this case we write the  leading term in $f$ as $-\frac{i}{2}\kappa e^{i\alpha}  Ae^{i\chi}  z$ for some complex number $Ae^{i\chi}$, and so that the  leading terms  in $v$ for large $r$ are 
 $-\frac{i}{2}\kappa e^{i\alpha}( \bar z + Ae^{i\chi} z)$.  Then one checks, using essentially the calculation leading to \eqref{Thetasy},  that the leading terms in  $\beta$  are 
 \bee
\label{betasy}
\beta \sim  2 \d \varphi +4\frac{A\sin(2\varphi +\chi)}{ 1 + A^2  +2A \cos (2\varphi +\chi)} \d \ln r.
\eee
 As already discussed in the paragraph following \eqref{Thetasy}, the presence of the term proportional to $ \d \ln r$ means that, for $A\neq 0$, the integral of $\beta$ cannot be given a meaning independently of the curve, even in the limit of large radius. However,  regularising by insisting on circular integration paths we observe
 \bee
 \label{mn1}
 \lim_{R\rightarrow \infty} \int_{C^R}\beta = 4\pi,
\eee
and hence 
\bee
4\pi (Q[w] + \Omega^\circ [w])  = 4\pi (n+1)  \quad \text{if} \; \; m=n+1.
\eee

\noindent  (iii) \;  $m < n+1$.\; Now the leading term in $v$ is $-\frac{i}{2}\kappa e^{i\alpha}\bar z$ and the  holomorphic term  is subleading  for large $r$. The integration of $\beta$ in the large $r$ limit is therefore a special case of the calculation in (ii), obtained by setting $A=0$.  This eliminates the term proportional to $\d \ln r$  in \eqref{betasy} and produces a limit independent of the chosen curve. We obtain 
 \bee
4\pi (Q[w] + \Omega [w])  = 4\pi (n+1)  \quad \text{if} \;\;  m<n+1,
\eee
which completes the proof.
\end{proof}

For the remainder of the paper  we focus on rational solutions of the form \eqref{findeg} and  set
\bee
\label{Ndef}
N= \text{max}(m,n+1).
\eee
A simple counting argument shows that  there is a $4N$ (real-)dimensional moduli space of rational maps of the form \eqref{ratf}. For $N=1$  (regularised energy $4\pi$), the 4-dimensional family of solutions is conveniently written as 
\bee
\label{N1fam}
v_1(z)=-\frac{i}{2}\kappa e^{i\alpha}\left(\bar z + az\right)+b \quad a,b \in \C. 
\eee
This includes the family \eqref{keyex} discussed in detail in the previous section  for $b=0$.
For $N=2$ (energy  or regularised energy $8\pi$), the 8-dimensional family of solutions can be written as 
\bee
\label{N2fam}
v_2(z)=-\frac{i}{2}\kappa e^{i\alpha}\bar z + \frac{az^2 +bz+c}{dz+e} , \quad a,b,c,d,e \in \C,  (a,b,c,d,e)\sim \lambda(a,b,c,d,e), \; \lambda \in \C^*,
\eee
where the equivalence relation removes the  redundant simultaneous rescaling of  numerator and denominator by the same non-zero complex number, and we need to require that (i)  $a$ and $d$ do not vanish simultaneously, (ii) $d$ and $e$ do not vanish simultaneously and (iii) the resultant of $(a,b,c,d,e)$ is non-vanishing to ensure that numerator and denominator do no have common factors \cite{MS}.
We will  discuss the form and energy distribution of some of these solutions in the next section.

\section{Solutions and their energy density}

The results obtained thus far add up to a simple recipe for constructing magnetisation fields $\bn$ which solve the Bogomol'nyi equation (and hence the variational equation \eqref{eomm}) out of two complex polynomials $p$ and $q$. Inserting these polynomials into the expression \eqref{findeg} for the complex field $v$,    and combining \eqref{vdef} and \eqref{invstereo} to write the magnetisation  $\bn$ as 
\bee
\label{nv}
n_1+in_2= \frac{2\bar{v} }{|v|^2+1}, \qquad n_3 = \frac{|v|^2-1}{|v|^2+1},
\eee
we obtain solutions of the Bogomol'nyi equation \eqref{alphabogon}.

The simplest solutions are obtained when the holomorphic contribution to $v$  vanishes, i.e., when $f=0$. We use them to illustrate the translation from our coordinates into the ones conventionally used in the discussion of magnetic skyrmions in the literature.  Translating   $v=-\frac{i}{2}\kappa e^{i\alpha}\bar  z$  into the magnetisation field via \eqref{nv}, and comparing with the hedgehog parametrisation \eqref{nangle},   we deduce
\bee
\label{hedgehog}
\theta = 2\tan^{-1}\left(\frac{2}{\kappa r}\right).
\eee
This yields the  Bloch  skyrmion for $\gamma=\frac{\pi}{2}$ (so $\alpha =0$) and the N\'eel skyrmion for $\gamma = 0$ (so $\alpha=\frac{\pi}{2}$)  in their standard form \cite{NT}, but with a particularly simple profile function interpolating between $\theta =\pi$ at  $r=0$ and $\theta =0$ at  $r=\infty$. The solutions \eqref{hedgehog} agree with  the  hedgehog solutions \eqref{hedgehoggen}   of the variational equations when $\mu=\kappa$.

For the remainder of this section we  set $ \kappa=1$, and study some example solutions of the Bogomol'nyi equation \eqref{alphabogon} in some detail.  We adopt the convention, widely used in the magnetic skyrmion literature,  to  refer to configurations of negative degree (like the Bloch and N\'eel solutions) as skyrmions, and to the configurations of positive degree as anti-skyrmions.

We have organised our discussion according to the integer $N$ defined in  \eqref{Ndef}, and begin with the $N=1$  family  \eqref{N1fam}. Of the four real parameters in the two  complex numbers $a$ and $b$, three can be understood in terms of the symmetry group of translations and rotations \eqref{rotact}. Rotations leave the basic skyrmion ($a=b=0$) invariant, but translations generate a shift in $b$. For the general  configuration \eqref{N1fam}, rotations  by $\sigma$ act  by mapping 
\bee
-\frac{i}{2}\kappa e^{i\alpha}\left(\bar z + az\right) + b \mapsto  -\frac{i}{2}\kappa e^{i\alpha} \left(\bar z + e^{-2i\sigma} a z\right) + e^{-i\sigma} b,
\label{rott}
\eee
so can be used to adjust the phase of $a$. 

This action  has an  elementary but interesting consequence  when $b=0$.  Since a rotation by  some angle $\sigma$ leaves the anti-holomorphic term invariant but rotates the phase of the linear  holomorphic term by $-2\sigma$,  configurations with $a\neq 0 $ and $b=0$ are mapped to themselves after a rotation by $\pi$.  This clearly generalises to configurations with a homogeneous holomorphic part  proportional to $z^n$, which  are mapped to themselves after a rotation by $2\pi/n$.

The only parameter in the  $N=1$ family \eqref{N1fam} which cannot be adjusted by a symmetry transformation is  the magnitude  $|a|$ of the complex coordinate $a$.  Varying it  leads to the most interesting deformation in  this  family.  We already know that for $a=0$ we have the basic hedgehog skyrmion.  According to our  formula  \eqref{Qformula},  we have degree $ -1$ configurations, i.e. skyrmions, for $|a|<1$ and  degree $1$ configurations, i.e. anti-skyrmions,  for $|a| > 1 $.  The interpolation between the two necessarily involves the  case $|a|= 1 $ where the degree is zero. As  observed in the discussion after \eqref{Qformula},  the stereographic coordinate $w$ has a pole along an entire  line  in this case. The magnetisation takes the value $\bn=(0,0,-1)^t$  and both the  potential $\frac 12 (1-n_3)^2$  and the energy density \eqref{endensity}  are maximal along  this line. We therefore call  $N=1$ configurations with $|a|=1$ line defects.

\begin{figure}[ht]
\centering
\includegraphics[width=0.3\textwidth] {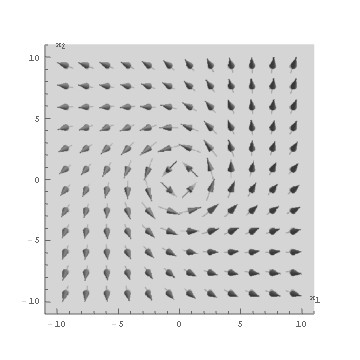} \hspace{1.0cm}
\includegraphics[width=0.3\textwidth]{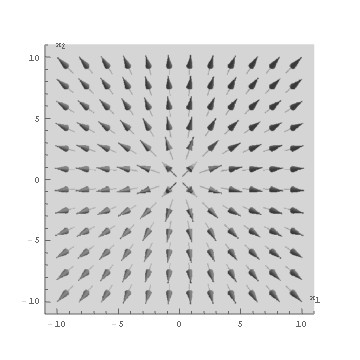} \\
\vspace{-0cm}
 \includegraphics[width=0.3\textwidth]{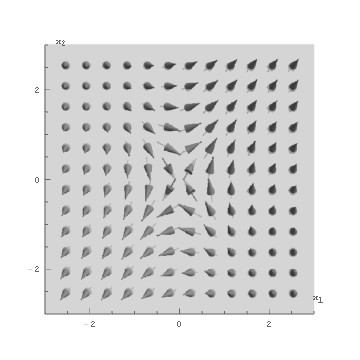} \hspace{1cm}
   \includegraphics[width=0.3\textwidth]{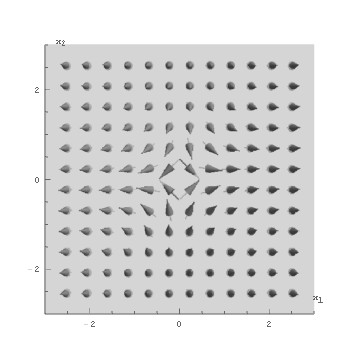}
   \vspace{0.2cm}
    \caption{Top from  left to right: Bloch skyrmion  $v=-\frac{i}{2}\bar{z}$ in the theory with $\alpha=0$ and   N\'eel skyrmion $v=\frac{1}{2}\bar{z}$ in the theory with $\alpha =\frac{\pi}{2}$. Bottom:  anti-skyrmions in the theory with $\alpha=0$, with   $v=-\frac{i}{2}\left(\bar{z}+ 4z\right)$ shown on the  left  and   $v=-\frac{i}{2}\left(\bar{z}+6iz  \right)$ on the right. Note that the  magnetisation of anti-skyrmions rotates oppositely to that of the skyrmions when one traverses a positively oriented circle around the centre, where $v=0$ or $\bn=(0,0,-1)^t$. Note also the different scale of the anti-skyrmion plots.   Size and orientation of anti-skyrmions  can be adjusted via the coefficient of $z$.}
    \label{fig1234}
\end{figure}

To sum up, as $|a|$ increases from zero 
 we deform a  hedgehog skyrmion  through a line defect into an anti-skyrmion. 
 In this process, the axisymmetric   energy distribution of a hedgehog  skyrmion  is  first squeezed and stretched into an elliptical shape and then into a  line at $|a|=1$. As $|a|$ is increased further    the energy distribution contracts again into an elliptical shape and shrinks; it never recovers the axisymmetry of the original skyrmion.  In Fig.~\ref{fig1234} we illustrate our discussion  by showing the magnetisation of the  basic Bloch and N\'eel skyrmions in our model, and  also the magnetisation of two anti-skyrmions  in the theory  with $\alpha =0$.


Next we turn our  attention to the family of solutions \eqref{N2fam} with $N=2$.  We have not fully explored the eight parameters in this family, of which three can again be accounted for by the symmetry operations of translation and rotation.  Here, we only exhibit two interesting phenomena and fix $\alpha=0$ for definiteness. The first is  the nonlinear superposition of skyrmions and anti-skyrmions in this model.  Configurations of the form   
\bee
v=-\frac{i}{2}\left(\bar{z}- \frac{ z^{n}}{R^{n-1}}\right), \qquad  R\in \R^{>0}, \quad n \in \Z^{>1} ,
\eee
have degree $n$,  but the functions $v$  have   $(n+2)$ zeros: one at the origin and $(n+1)$ zeros  at 
\bee
\label{vzero}
z_k=R e^{\frac{2\pi k i}{n+1}} , \quad k=0,\ldots, n.
\eee
The magnetisation takes the value $\bn=(0,0,-1)^t$ at the zeros, 
but inspection of the winding  shows that the configuration consists of a skyrmion  at the origin  surrounded by $(n+1)$ anti-skyrmions  at the locations \eqref{vzero}. The energy density \eqref{endensity} is peaked at the anti-skyrmion locations.   Such superpositions of 
 solitons and anti-solitons do not solve the Bogomol'nyi equations of the pure $O(3)$ sigma model,  which require maps to be either holomorphic or antiholomorphic. In our model they are possible, essentially because the number of zeros for functions of the form \eqref{findeg} can exceed the absolute value of the degree.  In fact,  determining the number of zeros of \eqref{findeg} (and  hence the number  of skyrmions and anti-skyrmions counted without sign) is an interesting  mathematical problem, see   \cite{KN} for  rigorous bounds and also \cite{BHJR,FKK} for further results and the application of this problem to gravitational lensing.

Continuing with $\alpha=0$,  we observe that rational solutions of the form 
\bee
\label{bag} 
v=-\frac{i}{2}\left(\bar{z} -  \frac{R^2}{z}\right), \quad  R\in \R^{>0}, 
\eee
  look like skyrmion bags \cite{FKATDS} or sacks \cite{RK}.  These particular bags  have  degree $Q=0$, take the vacuum value  $\bn=(0,0,1)^t$ at the origin and  the value $\bn=(0,0,-1)^t$ on a circle of radius $R$ centred at the origin. The energy density \eqref{endensity} is maximal on this circle.  The solutions \eqref{bag} are  axisymmetric and an example of what is sometimes called skyrmionium in the literature.  Clearly one can obtain more general bags by acting with translations on \eqref{bag}.  We note that the circles of zeros of functions like  \eqref{bag} also play a special role in the context of gravitational lensing, where they are  called  Einstein rings \cite{FKK}. 

\section{Conclusion}

In this paper   we introduced  models for magnetic skyrmions in the plane for which an infinite family  of analytical solutions can be given explicitly.  In our study 
we concentrated on the family of magnetic skyrmions determined by a rational holomorphic function according to \eqref{findeg}. We showed that, with a suitable regularisation in the case of linear growth in the holomorphic function at infinity,  the total energy 
takes  the  quantised values  $4\pi N$,  where $N$ is a positive integer which combines the degree with the  (possibly regularised) total vorticity of a configuration. This integer does not appear to have been studied in the literature on magnetic skyrmions, but our study suggests that it plays an important role. 

We also   determined  the  collective coordinates  or moduli of the solutions for given $N$, and  studied example configurations  for low values of $N$. They include remarkable deformations of Bloch and N\'eel skyrmions  into  line defects and anti-skyrmions.  Finally, we exhibited some of the  bag  and multi-anti-skyrmion configurations included in the family \eqref{findeg}. 

Magnetic skyrmions are sometimes also called chiral skyrmions because the DM interaction breaks reflection symmetry. This is evident in our solutions through the mandatory and fixed anti-holomorphic part (of  negative degree)  but the  optional holomorphic part (of  positive degree). It is interesting and somewhat unexpected that our models allows for   nonlinear superpositions  of skyrmions and anti-skyrmions in static configurations.  However, their roles are not symmetric. To obtain  perfect mirrors of our models one would need to replace the DM interaction with the one of opposite chirality \eqref{reflectedDM}. One checks that this would lead to a Bogomol'nyi equation which enforces a fixed holomorphic part, but allows for an arbitrary anti-holomorphic part.  This is consistent with the criterion derived  in \cite{genDMI}  for the  preference of   skyrmions over  anti-skyrmions depending on the chirality. 

The explicit family of solutions in our models and their  chiral twins should be studied further in order to obtain a systematic understanding of the types of defects they capture. One would also like to understand how these defects relate  to those observed in the various phases of generic models for magnetic skyrmions as discussed for example in \cite{phasediagram}.  

Mathematically, one  would like to  understand more precisely the class of maps from the plane to the sphere  for which the integrals defining the degree, total energy and total vortex strength  are  well-defined. 
It would also be important to ascertain if  the energy functional \eqref{alphaen1} is bounded below for a suitably defined class of configurations (not just our solutions).  In \cite{Melcher} it was  shown that in a closely related model with a pure Zeeman potential the  total energy is bounded below by a multiple of the  absolute value of the degree.  In our model, the energy is potentially unbounded below unless one imposes suitable  behaviour at spatial infinity.

The second order variational equation \eqref{eomm} and the Bogomol'nyi equation \eqref{holobog}  deserve further study.  One would like to know, for example, if there are other finite-energy solutions (possibly after suitable regularisation)  of the variational equation, not included in  our rational  family \eqref{findeg}.

Finally, future work  should    include the study of time evolution  and the effect of external fields 
on the solutions in  our model.

\vspace{0.5cm}

\noindent {\bf Acknowledgements} \, B B-S and CR acknowledge EPSRC-funded PhD studentships. BJS thanks Christof Melcher for correspondence and for pointing out reference \cite{DM}, and  Paul Sutcliffe for  helpful  comments regarding the total energy of line defects in our model.

\vspace{0.5cm}

\noindent {\bf Note added}  \, While this paper was under review,   
research was reported in the literature which has implications for the definition of energy in our models.   A modification of the energy expression for magnetic skyrmions by a boundary term already proposed for analytical reasons in  \cite{Melcher}  was generalised  in \cite{Schroers2} in the framework of gauged sigma models. Applied to the models discussed here, this modification amounts to subtracting  the total vorticity $\Omega_{\R^2}$ \eqref{EQOmD} from our energy \eqref{alphaen1} or, equivalently,  to  replacing  $\bn \cdot \nabla_{-\alpha}  \times \bn$ by  $(\bn-\be_3) \cdot \nabla_{-\alpha}  \times \bn$ in \eqref{alphaen1}.  The modification does not affect the variational equations or the Bogomol'nyi equations studied here.  The modified energy has the advantage of being  finite  without regularisation for all rational solutions  of the form \eqref{findeg}.  In fact,  for any such solution the modified energy is  equal to $4\pi Q$, where $Q$ is the degree, instead of the value $4\pi N$ for our energy (after regularisation if required).  On the  other hand, the   geometrical investigation reported in  \cite{Walton} shows  that  the  integer $N$ \eqref{Ndef} can  be interpreted as the equivariant degree of the rational solutions \eqref{findeg}, suggesting that the energy we used   here  is geometrically  natural.  It thus appears that there is a certain tension between the  energy expressions preferred from an analytical and a geometrical point of view.  Further work is required to resolve this.

\end{document}